%% file: mainv2.tex
\documentclass[prl,twocolumn,amsmath,amssymb,superscriptaddress, floatfix]{revtex4-2}
\usepackage{amsfonts}
\usepackage{graphicx}
\usepackage{grffile}
\usepackage{enumerate}
\usepackage{tensor}
\usepackage{mathrsfs}
\usepackage{comment}
\usepackage{amsthm} 
\usepackage{xcolor}
\usepackage{mathtools}
\usepackage{booktabs}
\usepackage{array}

\usepackage[english]{babel}
\usepackage{braket}
\usepackage{dcolumn}
\usepackage{bm}
\usepackage[hidelinks]{hyperref}
\usepackage{slashed}  
\usepackage{algorithm}
\usepackage{algpseudocode}

\newtheorem{theorem}{Theorem}

\begin{document}

\title{Quantum Memory Advantage from Contextuality}

\author{Shiroman Prakash}
\affiliation{Department of Physics and Computer Science, Dayalbagh Educational Institute, Agra, India}

\begin{abstract}
    Quantum contextuality is a vital non-classical resource, yet illuminating the precise mechanisms through which it enables unconditional computational advantages remains a challenge. We translate graph-theoretic formulations of contextuality into an unconditional quantum memory advantage for formal language recognition. We define a promise problem on an exclusivity graph $G$ where any classical finite automaton respecting exclusivity requires $N \ge \chi(G)$ memory states, whereas a QFA requires a memory of dimension $d = \xi(G)$. The gap between these bounds isolates a structural, information-theoretic incompatibility between classical and quantum descriptions that we term \textit{representational contextuality}. For Boolean orthogonality graphs, this exacts an exponential classical memory penalty ($d=\mathcal{O}(n)$ vs $N=2^{\Omega(n)}$). Finally, we demonstrate a sharp algorithmic phase transition: allowing the classical machine a finite confusability of mutually exclusive events reduces this exponential classical memory cost to $\mathcal{O}(n)$.
\end{abstract}

\maketitle

\section{Introduction}

Quantum contextuality seeks to characterize the non-classicality of quantum mechanics, proving that the measurement outcomes of physical observables cannot be explained by pre-existing, context-independent values~\cite{Bell1966, KochenSpecker1967}. While widely recognized as a vital resource for quantum information processing~\cite{acin2007device,Anders_2009,Raussendorf2013,acin2015combinatorial,nature,DPS2, Vega2017,Emeriau2022}, identifying precise operational mechanisms translating contextuality into an unconditional computational advantage remains an ongoing challenge. 

To better isolate these mechanisms, classical finite automata~\cite{kozen1997automata, hopcroft2007introduction, sipser2013introduction} and their quantum counterparts~\cite{KondacsWatrous1997, AmbainisFreivalds1998, MooreCrutchfield2000, BrodskyPippenger2002,AmbainisYakaryilmaz2021} provide an ideal theoretical playground. These sequential control protocols rely on a bounded internal state space to process information. The minimal size of this state space -- the \textit{state complexity} -- precisely quantifies the representational overhead incurred when a classical system attempts to simulate a quantum reality.

Here, we demonstrate that a structural prerequisite of contextuality -- the topological uncolorability of exclusivity graphs~\cite{Cameron2007, Cabello2008, CabelloSeveriniWinter2014, Ramanathan2014, Cabello2015} -- gives rise to a language recognition promise problem~\cite{Even1984PromiseProblems,Goldreich2006PromiseProblems} for which quantum finite automata (QFAs) possess an exponential memory advantage over exact classical automata. We show that adherence to physical exclusivity exacts a classical memory penalty equal to the integer gap between chromatic number and orthogonal rank, isolating an incompatibility between classical and quantum models we term \textbf{\textit{representational contextuality}}. Representational contextuality is a form of non-classicality directly tied to a computational resource: memory cost.

For the family of Boolean orthogonality graphs, this mismatch forces an exponential gap in classical versus quantum state complexity ($2^{\Omega(n)}$ vs $\mathcal O(n)$). However, we identify an algorithmic phase transition: permitting the classical automaton a finite statistical confusability between mutually exclusive outcomes beyond a critical threshold reduces the exponential classical state complexity to linear ($\mathcal{O}(n)$), providing an algorithmic analogue to the collapse of logical contextuality under unsharp measurements~\cite{meyer1999, kunjwal2015}.

\section{Representational contextuality and exclusivity graphs}

To connect contextuality to graph theory, we utilize Kochen-Specker configurations, formalized as \textit{exclusivity graphs} $G=(V,E)$~\cite{KCBS, Cameron2007, CabelloSeveriniWinter2014, Cabello2015}. Vertices $v_i \in V$ represent general experimental outcomes, and edges connect mutually exclusive outcomes. A quantum representation assigns a rank-1 projector $\hat{P}_i = \ket{v_i}\bra{v_i}$ to each vertex such that adjacent vertices map to orthogonal projectors. The minimum dimension $d$ required for a valid quantum representation is the \textit{orthogonal rank}, $\xi(G)$.

Conversely, to evaluate the classical representational cost, we employ Spekkens' ontological model framework~\cite{Spekkens_2005}, which defines a physical system as an epistemic probability distribution over an underlying space of ontic states, $\Lambda$. Each vertex $v_i \in V$ represents a distinct measurement event. Because adjacent vertices denote mutually exclusive outcomes, no single ontic state $\lambda \in \Lambda$ can simultaneously yield a non-zero probability for both. Therefore, for any given ontic state, the collection of measurement events it can support with non-zero probability must constitute an independent set of $G$. 

We now demand that a valid model must be capable of realizing every experimental outcome; thus, for each vertex, there must exist at least one ontic state in $\Lambda$ that allows its observation. Since every vertex must be assigned to at least one such ontic state, the available ontic states function as a complete independent set cover of the graph; such mapping is mathematically identical to a valid vertex coloring. The minimum size of the classical state space $N = |\Lambda|$ is therefore bounded by the integer chromatic number:
\begin{equation}
N \ge \chi(G).
\end{equation}

In standard statistical frameworks, an essential condition for state-independent contextuality is that the fractional chromatic number exceeds the orthogonal rank, $\chi_f(G) > \xi(G)$~\cite{Ramanathan2014, Cabello2015}, which ensures that the maximally mixed state violates a statistical contextuality inequality. Because the integer chromatic number inherently bounds its fractional counterpart ($\chi \ge \chi_f$), the integer gap $\chi(G) > \xi(G)$ serves as an even broader necessary-but-not sufficient condition for state-independent contextuality. 

Many graphs exist for which state-independent contextuality is absent (as evidenced by $\chi_f \le \xi$), yet the integer memory gap persists ($\chi > \xi$). A concrete demonstration is the 18-vertex graph join of the 13-vertex Yu-Oh graph~\cite{Yu_Oh_2012} and a 5-cycle~\cite{KCBS}, shown in Fig.~\ref{fig:counterexample}. Because $\chi_f < \xi$ for this graph, the Ramanathan-Horodecki criterion precludes universal statistical violation of noncontextual inequalities for arbitrary quantum states. Nevertheless, because $\chi > \xi$, an unavoidable increase in the size of state-space occurs if one attempts to simulate the graph's exclusivity relations using an ontological model.
We therefore formalize this integer gap as a distinct form of contextuality, which we term \textit{representational contextuality}. 

\begin{figure}[!htbp]
    \centering
    \includegraphics[width=0.7\linewidth]{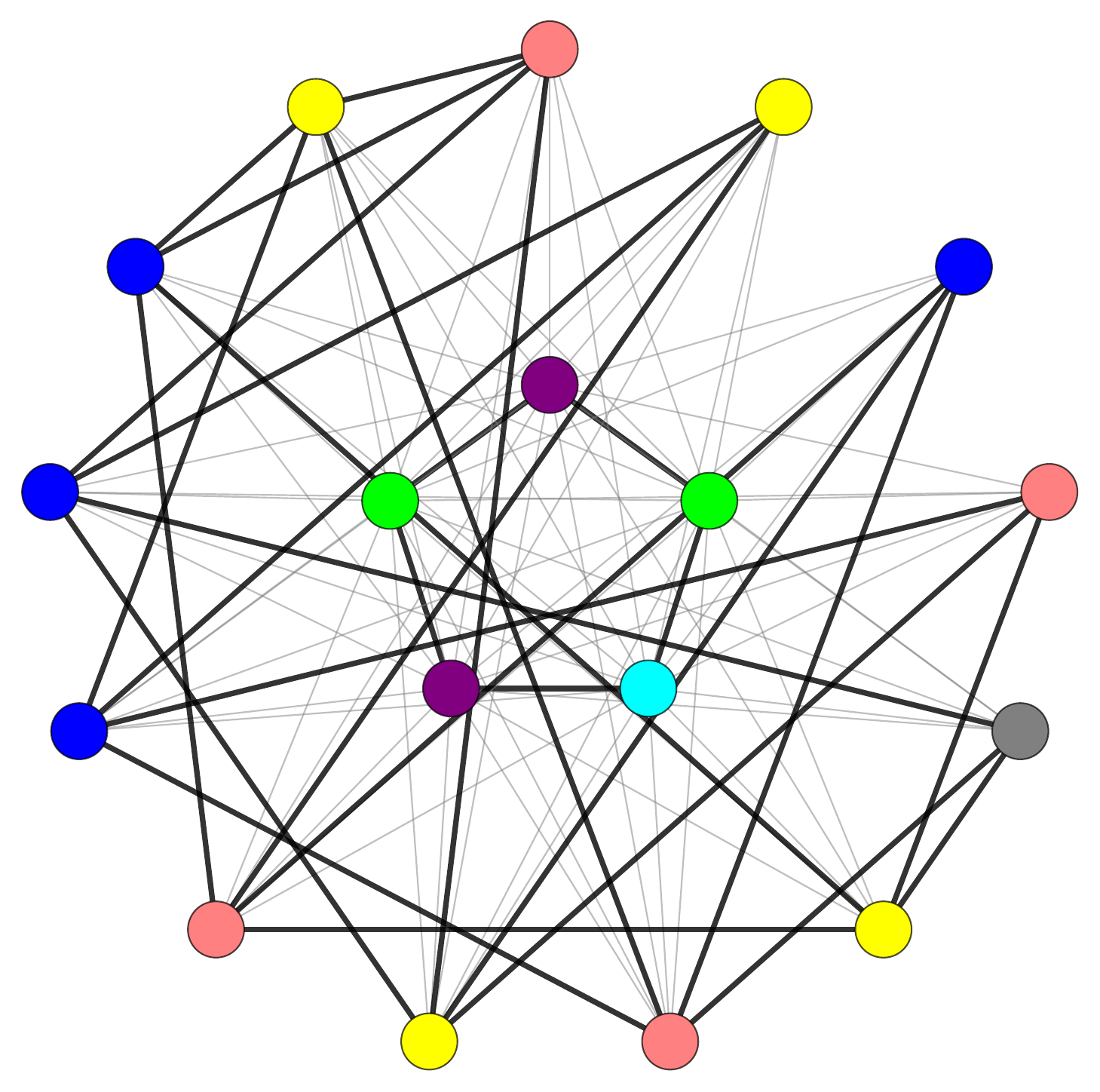}
 \caption{The 18-vertex graph join of the Yu-Oh graph (outer ring) and $C_5$ (inner ring), with colored vertices illustrating a valid $7$-coloring (representing the chromatic number $\chi=7$). Dark edges denote connections within the original subgraphs, while light edges represent the join operations. With invariants $\chi=7$, $\xi=6$, and $\chi_f \approx 5.68$, this graph cannot exhibit state-independent statistical contextuality ($\chi_f < \xi$). However, it exhibits representational contextuality, as the  minimum cardinality of its classical ontic state space exceeds the dimension of its optimal quantum representation ($\chi > \xi$).}
    \label{fig:counterexample}
\end{figure}

Let us also discuss the relation between representational contextuality and state-dependent contextuality. Under the Cabello-Severini-Winter framework~\cite{CabelloSeveriniWinter2014}, a graph avoids state-dependent contextuality for all observables only if its Lov\'asz number equals its weighted independence number, $\vartheta(G, w) = \alpha(G, w)$ for all weightings $w \ge 0$. By the Gr\"otschel--Lov\'asz--Schrijver theorem~\cite{GLS1981, knuth1993sandwichtheorem}, this equivalence holds if and only if $G$ is a perfect graph, which requires the chromatic number to equal the clique number, $\chi(G) = \omega(G)$~\cite{Berge1961}. Since any orthogonal representation of a clique requires mutually orthogonal projectors, the dimension is bounded by $\xi(G) \ge \omega(G)$. Therefore, any graph lacking state-dependent contextuality must satisfy $\chi(G) \le \xi(G)$. 

By contraposition, the presence of representational contextuality ($\chi(G) > \xi(G)$) is a sufficient condition to guarantee a state-dependent contextuality violation. Conversely, it is not a necessary condition; the 5-cycle graph possesses no representational gap ($\chi = \xi = 3$), yet famously supports state-dependent contextuality~\cite{KCBS}. 

To summarize, we establish a strict hierarchy among the sets $\mathcal{G}_{\text{SIC}}$, $\mathcal{G}_{\text{RC}}$, and $\mathcal{G}_{\text{SDC}}$ of exclusivity graphs supporting state-independent, representational, and state-dependent contextuality, respectively: 
\begin{equation}
\mathcal{G}_{\text{SIC}} \subset \mathcal{G}_{\text{RC}} \subset \mathcal{G}_{\text{SDC}}.
\end{equation}
This containment chain underscores that representational contextuality isolates an intrinsically memory-theoretic manifestation of non-classicality, that must be distinguished from both state-dependent and state-independent statistical contextuality.

\section{The Kochen-Specker Problem} 

We translate this structural gap into an unconditional state complexity advantage by defining the Kochen-Specker Problem (KSP) as a language recognition task on $G=(V,E)$. Formal definitions of the automata models and the promise problem framework are reviewed in the Supplemental Material.

The input alphabet is the vertex set, $\Sigma = V$. The automaton is fed strings of a constant length $L=2$, representing a sequential pair of experimental outcomes $w = v_i v_j$. This restriction to input strings of constant length is inspired by unconditional advantages in shallow-depth circuits derived from non-locality~\cite{Bravyi_2018}, and isolates the structural advantage of the quantum state space from trivial artifacts of accumulated continuous unitary rotations over asymptotically long strings.

The input strings obey a promise: the two symbols must either be identical ($v_i = v_j$) or mutually exclusive ($(v_i, v_j) \in E$). An automaton evaluates this language under two metrics:
\begin{itemize}
    \item \textbf{Completeness Error ($\epsilon_{\rm c}$)}: The maximum probability of incorrectly rejecting a valid string ($v_i = v_j$).
    \item \textbf{Soundness Error ($\epsilon_{\rm s}$)}: The maximum probability of falsely accepting an invalid string ($(v_i, v_j) \in E$).
\end{itemize}
Crucially, the soundness error $\epsilon_{\rm s}$ is the exact computational analogue of physical exclusivity. 

\section{State-Complexity of Sound Automata}

We first evaluate the memory required to resolve the KSP under absolute soundness ($\epsilon_{\text{s}} = 0$), demonstrating an unconditional memory penalty for classical automata that perfectly respect exclusivity.

\begin{theorem}
Any classical finite automaton (including a probabilistic finite automaton, PFA) that solves the Kochen-Specker Problem for $G=(V,E)$ with absolute soundness ($\epsilon_{\text{s}} = 0$) and non-zero completeness must possess an internal state space of size $N \ge \chi(G)$.
\label{thm:classical_bound}
\end{theorem}
\begin{proof}[Proof sketch] We regard the $N$ internal memory states of the PFA as an ontic state space $\Lambda$. Processing a string $w=uv$ splits the PFA's operation: reading the first symbol $u$ prepares an epistemic distribution $\mu_u(\lambda)$, while the transition and acceptance rules for the second symbol $v$ define an indicator function $\xi_v(\lambda) \in [0, 1]$. The total probability of accepting the string is
\begin{equation}
P(\text{accept } uv) = \sum_{\lambda \in \Lambda} \mu_u(\lambda) \xi_v(\lambda).
\end{equation}
Absolute soundness ($\epsilon_{\text{s}} = 0$) demands $P(\text{accept } uv) = 0$ for all mutually exclusive inputs $(u, v) \in E$. As established in Section II, satisfying this condition for all adjacent vertices simultaneously requires the ontological supports of mutually exclusive preparations to be disjoint. Mapping each vertex to an internal state within its support thus yields a valid vertex coloring, establishing $N \ge \chi(G)$. A more detailed derivation is provided in the Supplemental Material.
\end{proof}

In contrast, a quantum machine bypasses this partitioning constraint by encoding exclusivity through vector orthogonality.

\begin{theorem}
There exists a measure-once QFA that solves the KSP using a quantum memory of dimension $d = \xi(G)$.
\label{thm:QFA_bound}
\end{theorem}
\begin{proof}
Construct a minimal-rank orthogonal representation mapping each $v_i \in V$ to a unit vector $\ket{v_i} \in \mathbb{C}^{\xi(G)}$. Let the initial memory state be a unit vector $\ket{\psi_0} \in \mathbb{C}^{\xi(G)}$ not parallel to any $\ket{v_i}$. Choose the (arbitrary) overall phase of each $\ket{v_i}$ so that $\braket{\psi_0|v_i} \in \mathbb{R}$. 

The acceptance measurement is the projector $\hat{P}_{\text{accept}} = \ket{\psi_0}\bra{\psi_0}$. For each symbol $v \in V$, the transition unitary is the Householder reflection $U_v = I - 2\ket{v'}\bra{v'}$ about the hyperplane orthogonal to $\ket{v'} = (\ket{\psi_0} - \ket{v}) / \|\ket{\psi_0} - \ket{v}\|$, which maps $U_v \ket{\psi_0} = \ket{v}$.
\begin{enumerate}
    \item \textbf{Valid Strings ($v_A v_A$):} The pre-measurement state is $U_{v_A}^2 \ket{\psi_0} = \ket{\psi_0}$. The acceptance probability is $\bra{\psi_0}\hat{P}_{\text{accept}}\ket{\psi_0} = 1$.
    \item \textbf{Invalid Strings ($v_A v_B$):} The memory state is $U_{v_B} U_{v_A}\ket{\psi_0} = U_{v_B} \ket{v_A}$. By Hermiticity, the acceptance amplitude is $\bra{\psi_0} U_{v_B} \ket{v_A} = \braket{v_B|v_A} = 0$.
\end{enumerate}
The QFA solves the KSP using exactly $d = \xi(G)$ dimensions.
\end{proof}

The memory requirements of these automata form a hierarchy governed by graph-theoretic invariants. A QFA maps adjacent inputs to orthogonal pure states; its quantum memory dimension is therefore the orthogonal rank $\xi(G)$. In contrast, the minimum state complexity of a deterministic classical automaton is constrained by the chromatic number $\chi(G)$. 

Although shared randomness cannot reduce the cardinality of state space, transitioning to entropic measures~\cite{Korner_1973, Csiszar_1990, Simonyi_1995} reveals an effective reduction in classical cost. As proven in the Supplemental Material, the worst-case entropic complexity of a sound classical automaton with shared randomness is bounded by the fractional chromatic number $\chi_f(G)$. Physically, this fractional bound dictates the theoretical minimum average memory per run required to store a historical sequence of the sampled hidden variables after optimal data compression; consequently, the gap $\chi > \chi_f$ implies an ontological model that demands a strictly larger repertoire of distinct physical states than its compressed information-theoretic footprint would otherwise necessitate. When classical machines are granted shared randomness and state complexity is measured via this Shannon entropy, the quantum memory advantage persists, but only for those exclusivity graphs where $\chi_f(G) > \xi(G)$.

For the family of $n$-bit Boolean-orthogonality graphs (detailed in the Supplemental Material), this representational mismatch scales exponentially~\cite{FranklRodl1987, Buhrman1998, Cameron2007}. The quantum machine requires a memory dimension scaling linearly, $d = \mathcal{O}(n)$, whereas any absolutely sound classical simulation faces an exponential memory penalty, $N \geq \chi_f \geq 2^{\Omega(n)}$.

\section{Phase Transition in State Complexity for Bounded-Error Automata}

To isolate the mechanism driving this separation, we permit the classical PFA to operate with a non-zero completeness and soundness error. 
\begin{theorem}
If a classical PFA evaluates the KSP on a graph $G$ with completeness error $\epsilon_{\text{c}} \ge 0$ and soundness error $\epsilon_{\text{s}} > 0$, its classical state complexity is lower-bounded by
$$N \ge  
\chi(G), $$ provided that the soundness error is below the critical threshold,  $$\epsilon_{\text{s}} \le \frac{(1 - \epsilon_{\text{c}})^2}{4\chi(G)}.$$
\label{thm:unsound}
\end{theorem}

\begin{proof}[Proof sketch]By relaxing strict soundness, the internal states of the automaton become noisy indicator functions that overlap over adjacent edges; our proof, given in the Supplemental Material, then determines how much error the system can absorb before the exclusivity graph's independent set structure breaks down.
\end{proof}

This mathematical structure reveals an operational asymmetry. If we require zero soundness error ($\epsilon_{\text{s}} = 0$), the classical memory cost remains bounded below by $\chi(G)$ even if the completeness error approaches 1. The memory overhead is dictated entirely by the logical constraint of exclusivity.

In quantum foundations, the Meyer-Kent-Clifton (MKC) theorem demonstrates that logical Kochen-Specker proofs are  fragile: finite measurement precision  allows classical hidden-variable models to reproduce the target quantum statistics~\cite{meyer1999, kent1999, mermin1999, clifton2000, appleby2002, kent2004, garola2004}. To address this vulnerability, Kunjwal and Spekkens established finite noise thresholds for the robust violation of statistical noncontextuality inequalities~\cite{kunjwal2015}. Theorem~\ref{thm:unsound} provides an algorithmic analogue to these foundational results, extending the noise-robustness framework from statistical correlations to representational contextuality. While the exact logical Kochen-Specker contradiction collapses under infinitesimal experimental imperfections, our topological memory penalty survives finite noise. A classical model can bypass the $\chi(G)$ bound only when permitted to confuse exclusive events past a critical, graph-dependent threshold:
\begin{equation}
    \epsilon_{\text{s}} > \frac{1}{4\chi(G)}.
\end{equation}

Once the acceptable soundness error exceeds this threshold, a classical automaton can exploit this error-tolerance to implement randomized fingerprinting protocols (see Supplemental Material). Specifically, executing a random walk on a Ramanujan expander graph~\cite{Alon_1986, Lubotzky_1988, Margulis_1988, Hoory_2006} reduces the classical state complexity from $2^{\Omega(n)}$ to $\mathcal{O}(n)$, eliminating the quantum memory advantage. 

Notably, however, the asymmetry of Theorem~\ref{thm:unsound} identifies a specific noise regime where this classical complexity transition is avoided. Under a quantum erasure channel, such as pure particle loss, a state either survives or is completely lost. This induces a completeness error ($\epsilon_{\rm{c}} = p > 0$) but preserves a soundness error of exactly zero ($\epsilon_{\rm{s}} = 0$). While practical implementations of the QFA will encounter a mixture of noise sources, Theorem~\ref{thm:unsound} demonstrates that in purely erasure-dominated settings, the classical memory cost remains bounded from below by $\chi(G)$, preserving the exponential memory advantage.

Beyond the theoretical preservation of this advantage under erasure, the experimental implementation of the $L=2$ QFA protocol itself is highly noise-resilient. The protocol maintains a constant, $\mathcal{O}(1)$ operational noise threshold independent of graph size for coherent and incoherent noise (see Supplemental Material for details).

\section{Conclusion and outlook}
Representational contextuality is a form of nonclassicality tied, by construction, to a quantifiable computational resource: memory. We have shown that it imposes an exponential memory penalty on any classical simulation that strictly enforces the physical exclusivity of incompatible outcomes. This separation is governed by a phase transition in classical state complexity that exhibits a distinct noise asymmetry: it survives completeness errors, yet vanishes when unsharp measurements introduce soundness (i.e., confusability) errors.

While conventional statistical contextuality is well suited for designing empirical tests to rule out classical noncontextual theories, representational contextuality characterizes a fundamental information-theoretic cost that must be paid to even attempt a classical noncontextual description of nature. Distinct from state-dependent or state-independent statistical contextuality, it defines the exact memory overhead required whenever any generalized physical theory, whose observables respect the topological structure of an exclusivity graph, is classically simulated using single-shot evaluations. Looking ahead, a direction for future work will be to extend these single-shot topological constraints into sequential measurement regimes, potentially bridging our framework with temporal memory bounds derived via Mealy machines~\cite{Kleinmann_2011, trandafir2025memorycostquantumcontextuality}.

From a practical standpoint, the QFA solution to the KSP provides a non-trivial benchmark for experimental quantum platforms~\cite{Kirchmair_2009, Tian2019, mert2022implementing, cardoso2024implementing}. Because the required quantum dimension scales linearly with the orthogonal rank ($d = \xi$), executing the QFA requires minimal hardware. For example, the 60-vertex graph of Waegell and Aravind~\cite{Waegell_2010} requires only a single 4-level qudit or two qubits (see Supplemental Material for details). Scaling this to a modest processor with $\sim 10$ qubits ($d=1024$) is sufficient to evaluate Boolean orthogonality graphs with an exponentially large chromatic number, ensuring a theoretical quantum advantage. However, while the QFA itself is highly noise-robust, demonstrating an advantage over classical fingerprinting requires maintaining a soundness error below the critical graph-dependent threshold. Achieving this, perhaps using error correction or in a system whose noise is entirely erasure-dominated, appears to be a challenging yet attainable benchmark for demonstrating definitive quantum memory advantage in early fault-tolerant architectures.

\begin{acknowledgments}
\textbf{Acknowledgments.}
The author thanks Prof. Prem Saran Satsangi for guidance and inspiration. The author also thanks R. Loganayagam and the International Centre for Theoretical Sciences (ICTS) for hospitality when part of this work was completed, including the Discussion Meeting: A Hundred Years of Quantum Mechanics (ICTS/qm1002025/01). This work was supported in part by the DST-SERB grant (CRG/2021/009137).
\end{acknowledgments}

\bibliography{contextuality}


\clearpage
\begin{center}
  \textbf{\large Supplemental Material: Quantum Memory Advantage from Contextuality}
\end{center}

\setcounter{equation}{0}
\setcounter{figure}{0}
\setcounter{table}{0}
\setcounter{page}{1}
\setcounter{section}{0}
\setcounter{theorem}{0}
\makeatletter

\renewcommand{\theequation}{S\arabic{equation}}
\renewcommand{\thefigure}{S\arabic{figure}}
\renewcommand{\thetable}{S\arabic{table}}
\renewcommand{\thesection}{S\arabic{section}}
\renewcommand{\thetheorem}{S\arabic{theorem}}

\input{supplemental-v2}

\end{document}

%% file: supplemental-v2.tex
\section{Mathematical Proofs for Classical Finite Automata}
\label{supp:probabilistic_automata}

\subsection{1. Review of finite automata}
\textbf{Deterministic finite automata (DFA).} A DFA models sequential computation under memory constraints, and must accept or reject an input string based on whether it belongs to a specified formal language. Formally, a DFA is defined by the 5-tuple $$M_c=(S,\Sigma,\delta,s_0,S_{\rm accept}),$$ where $S$ is a finite set of discrete internal states, $\Sigma$ is the input alphabet, $\delta: S \times \Sigma \to S$ is the deterministic transition function, $s_0 \in S$ is the initial state, and $S_{\rm accept} \subseteq S$ specifies the subset of accepting states. The state complexity, or memory cost, of a classical automaton is measured by the total number of distinct states $|S|$ required to correctly resolve the language.

\textbf{Probabilistic finite automata (PFA).} A PFA generalizes the DFA by permitting stochastic transitions, modeling classical systems that utilize probabilistic memory. Formally, a 1-way measure-once PFA is defined by the 5-tuple $$M_p=(S, \Sigma, \{T_\sigma\}_{\sigma \in \Sigma}, \mathbf{p}_0, \mathbf{f}_{\rm accept}).$$ Here, $S$ is a finite set of discrete states with $N=|S|$ representing the memory cost, $\Sigma$ is the input alphabet, and $\mathbf{p}_0$ is an $N$-dimensional stochastic vector representing the initial probability distribution over $S$. For each input symbol $\sigma \in \Sigma$, $T_\sigma$ is an $N \times N$ column-stochastic transition matrix. Finally, $\mathbf{f}_{\rm accept}$ is an $N$-dimensional binary indicator vector (where $f_i \in \{0,1\}$) specifying the accepting states. For an input word $w=\sigma_1\ldots\sigma_k$, the probability distribution of the memory prior to measurement is given by $\mathbf{p}_w = T_{\sigma_k}\ldots T_{\sigma_1}\mathbf{p}_0$. The total operational probability that the string is accepted is evaluated as $\mathbf{f}_{\rm accept}^T \mathbf{p}_w$.

\textbf{Quantum finite automata. (QFA)} We will demonstrate quantum memory advantage using a 1-way measure-once QFA~\cite{MooreCrutchfield2000,AmbainisYakaryilmaz2021}, defined by the 5-tuple $$M_q=(\mathcal{H}, \Sigma,\{U_\sigma\}_{\sigma \in \Sigma},\ket{\psi_0},\hat{P}_{\rm accept}).$$ The quantum memory is a $d$-dimensional Hilbert space $\mathcal{H}$ initialized in a pure state $\ket{\psi_0}$. As it reads each input symbol $\sigma \in \Sigma$, the QFA applies a unitary evolution $U_\sigma$ to the memory, and measures $\hat{P}_{\rm accept}$ when the string terminates. For an input word $w=\sigma_1\ldots\sigma_k$, the physical state prior to measurement is $\ket{\psi_w} = U_{\sigma_k}\ldots U_{\sigma_1}\ket{\psi_0}$; the string is accepted with probability $\bra{\psi_w} \hat{P}_{\rm accept} \ket{\psi_w}$. The quantum memory cost is defined by the dimension $d = \dim(\mathcal{H})$.

\textbf{Promise problems.} A promise problem is a language recognition task where input strings are guaranteed to belong to a restricted subset $P \subseteq \Sigma^*$~\cite{Even1984PromiseProblems,Goldreich2006PromiseProblems}. The automaton must only distinguish valid strings belonging to a target language $L \subset P$ from invalid strings in $P \setminus L$. 

It is also conventional to allow automata to succeed with a finite probability of error. Let $P_{\rm valid}=1-\epsilon_{\rm c}$ and $P_{\rm invalid}=\epsilon_{\rm s}$ be the probabilities of accepting valid and invalid input strings, respectively. An automaton is \textit{sound} if $P_{\rm invalid}=0$ and \textit{complete} if $P_{\rm valid}=1$. The automaton is said to function with bounded-error if the acceptance gap $\Delta P=P_{\rm valid}-P_{\rm invalid} >0$~\cite{Rabin1963,AmbainisFreivalds1998}.

\subsection{2. State complexity of DFA solution to KSP}
We first establish that any classical DFA that solves the KSP induces a valid vertex coloring on the exclusivity graph $G$.

\begin{theorem}Any classical DFA that correctly solves the KSP for an exclusivity graph $G=(V,E)$ must possess at least $\chi(G)$ internal states, where $\chi(G)$ is the chromatic number of $G$.\end{theorem}
\begin{proof} Let $M_c=(S, \Sigma, \delta, s_0, S_{\rm accept})$ be a DFA that solves the KSP. Assume for contradiction that the state space satisfies $|S| < \chi(G)$. Let $v$ be the first symbol in the input string, and let $\delta(s_0,v)$ denote the state the DFA enters after reading $v$. The transition map $v \mapsto \delta(s_0,v)$ assigns an internal memory state to each vertex $v \in V$. Because $|S|< \chi(G)$, this mapping cannot constitute a valid vertex coloring. 

By the pigeonhole principle, there must exist at least two adjacent vertices, $v_A$ and $v_B$ with $(v_A, v_B) \in E$, that map to the identical internal state: $\delta(s_0, v_A)=\delta(s_0, v_B)=s_{\rm collision}$. 
For the valid input string $v_A v_A$, the DFA transitions to $s_{\rm collision}$ and then to $\delta(s_{\rm collision}, v_A)$. Correct acceptance requires $\delta(s_{\rm collision},v_A) \in S_{\rm accept}$. 
Conversely, for the invalid string $v_B v_A$, the DFA transitions to $s_{\rm collision}$ and then to $\delta(s_{\rm collision}, v_A)$. Correct rejection requires $\delta(s_{\rm collision},v_A) \notin S_{\rm accept}$. This is a direct contradiction. Hence, $|S| \geq \chi(G)$. \end{proof}

\subsection{3. PFAs and Ontological Models}
We formalize the operational mechanics of a classical probabilistic finite automaton (PFA) by mapping it onto the framework of classical ontological models~\cite{Spekkens_2005}. Let $S = \{s_1, s_2, \ldots, s_N\}$ be the internal state space of the PFA, which we regard as a set of ontic states $\Lambda$. Upon reading a vertex symbol $u$, the machine transitions into an epistemic state preparation $\mu_u(\lambda)$, where $\sum_{\lambda \in \Lambda} \mu_u(\lambda) = 1$. Upon reading the second symbol $v$, the combined transition and termination rules act as a measurement defined by an indicator function $\xi_v(\lambda) \in [0, 1]$. The probability of accepting the string $uv$ is:
\begin{equation}
P(\text{accept } uv) = \sum_{\lambda \in \Lambda} \mu_u(\lambda) \xi_v(\lambda).
\end{equation}

For some specified completeness error $\epsilon_{\rm c}$ and soundness error $\epsilon_{\rm s}$, the PFA's output must satisfy:
\begin{align}
\text{\textbf{Acceptance:}} \quad & \sum_{\lambda \in \Lambda} \mu_v(\lambda) \xi_v(\lambda) \ge 1 - \epsilon_{\rm c} \quad \forall v \in V, \label{eq:acceptance} \\
\text{\textbf{Rejection:}} \quad & \sum_{\lambda \in \Lambda} \mu_u(\lambda) \xi_v(\lambda) \le \epsilon_{\rm s} \quad \forall (u,v) \in E. \label{eq:rejection}
\end{align}

\subsection{4. Sound PFA Complexity (Proof of Theorem 1)}
We first consider a PFA that is permitted an imperfect completeness ($1>\epsilon_{\rm c} \ge 0$) but is perfectly sound ($\epsilon_{\rm s} = 0$), and provide a proof of Theorem 1 in the main text. The rejection constraint Eq.~\eqref{eq:rejection} requires that for any pair of adjacent vertices $(u, v) \in E$, $\sum_{\lambda \in \Lambda} \mu_u(\lambda) \xi_v(\lambda) = 0$. Because $\mu_u(\lambda)$ and $\xi_v(\lambda)$ are non-negative, this requires:
\begin{equation}
\mu_u(\lambda) \xi_v(\lambda) = 0 \quad \forall \lambda \in \Lambda, \quad \forall (u,v) \in E. \label{eq:soundness_product_strict}
\end{equation}

For each ontic state $\lambda \in \Lambda$, we define its active joint support $J_\lambda \subseteq V$ as:
\begin{equation}
J_\lambda = \{ v \in V \mid \mu_v(\lambda) > 0 \text{ and } \xi_v(\lambda) > 0 \}.
\end{equation}
$J_\lambda$ must be an independent set of $G$. If there existed adjacent vertices $u, v \in J_\lambda$, then $\mu_u(\lambda)\xi_v(\lambda) > 0$, violating Eq.~\eqref{eq:soundness_product_strict}. 

As long as the completeness requirement, Eq.~\eqref{eq:acceptance}, is non-trivial ($\epsilon_{\rm c} < 1$), the sum for any identical sequence must satisfy $$\sum_{\lambda \in \Lambda} \mu_v(\lambda) \xi_v(\lambda) \geq 1-\epsilon_{\rm c}>0.$$ This requires that every vertex $v \in V$ belongs to at least one active joint support set $J_{\lambda^*}$. Consequently, the collection of independent sets $\{J_\lambda\}_{\lambda \in \Lambda}$ forms a complete cover of $V(G)$. By definition, the minimum number of independent sets required to cover a graph is $\chi(G)$, yielding:
\begin{equation}
N_{\text{sound PFA}} \ge \chi(G),
\end{equation}
proving Theorem 1 of the main text.

\subsection{5. Weakly Unsound PFAs (Proof of Theorem 3)}
We now relax the soundness constraint and analyze a weakly unsound PFA where $\epsilon_{\rm s} > 0$. We introduce a continuous splitting parameter $\alpha \in (0, 1]$ and define a subset of vertices $I_\lambda(\alpha)$ for each state $\lambda$:
\begin{equation}
I_\lambda(\alpha) = \left\{ u \in V(G) : \mu_u(\lambda) > \epsilon_{\rm s}^\alpha \quad \text{and} \quad \xi_u(\lambda) > \epsilon_{\rm s}^{1-\alpha} \right\}
\end{equation}

We now establish that $I_\lambda(\alpha)$ is an independent set of $G$: Assume for contradiction that $u, v \in I_\lambda(\alpha)$ are adjacent, $(u, v) \in E$. By definition, $\mu_u(\lambda) > \epsilon_{\rm s}^\alpha$ and $\xi_v(\lambda) > \epsilon_{\rm s}^{1-\alpha}$. The total operational probability of accepting the adjacent string $uv$ satisfies:
\begin{equation}
P(\text{accept } uv) \ge \mu_u(\lambda) \xi_v(\lambda) > \left(\epsilon_{\rm s}^\alpha\right) \left(\epsilon_{\rm s}^{1-\alpha}\right) = \epsilon_{\rm s}.
\end{equation}
This directly violates the soundness rejection constraint $P(\text{accept } uv) \le \epsilon_{\rm s}$, proving that $I_\lambda(\alpha)$ is a valid independent set.

For any vertex $u$, we partition the ontic space $\Lambda$ into three disjoint sets relative to the threshold constraints:
\begin{align}
\Lambda_{\rm wp}^{(u)} &= \{ \lambda \in \Lambda \mid \mu_u(\lambda) \le \epsilon_{\rm s}^\alpha \}, \\
\Lambda_{\rm wm}^{(u)} &= \{ \lambda \in \Lambda \mid \mu_u(\lambda) > \epsilon_{\rm s}^\alpha, \text{ and } \xi_u(\lambda) \le \epsilon_{\rm s}^{1-\alpha} \}, \\
\Lambda_{\rm pass}^{(u)} &= \{ \lambda \in \Lambda \mid u \in I_\lambda(\alpha) \}.
\end{align}

We bound the maximum possible contribution of $\Lambda_{\rm wp}$ and $\Lambda_{\rm wm}$ to the completeness sum:
\begin{equation}
\sum_{\lambda \in \Lambda_{\rm wp}} \mu_u(\lambda) \xi_u(\lambda) \le \sum_{\lambda \in \Lambda_{\rm wp}} \left( \epsilon_{\rm s}^\alpha \right) (1) \le N\epsilon_{\rm s}^\alpha,
\end{equation}
\begin{equation}
\sum_{\lambda \in \Lambda_{\rm wm}} \mu_u(\lambda) \xi_u(\lambda) \le \epsilon_{\rm s}^{1-\alpha} \sum_{\lambda \in \Lambda_{\rm wm}} \mu_u(\lambda) \le \epsilon_{\rm s}^{1-\alpha}.
\end{equation}
Subtracting these bounds from the total required completeness yields:
\begin{equation}
\sum_{\lambda : u \in I_\lambda(\alpha)} \mu_u(\lambda) \xi_u(\lambda) \ge 1 - \epsilon_{\rm c} - N\epsilon_{\rm s}^\alpha - \epsilon_{\rm s}^{1-\alpha}. \label{eq:weight-in-set}
\end{equation}

We analyze the graph coverage under two scenarios:
\begin{itemize}
    \item \textbf{Case 1 (Incomplete Coverage):} If there exists a vertex $u$ not contained in any independent set $I_\lambda(\alpha)$, the left side of Eq.~\eqref{eq:weight-in-set} evaluates to $0$. Rearranging the inequality yields:
    \begin{equation}
    N \ge \frac{1 - \epsilon_{\rm c} - \epsilon_{\rm s}^{1-\alpha}}{\epsilon_{\rm s}^\alpha}.
    \end{equation}
    \item \textbf{Case 2 (Complete Coverage):} If every vertex $u \in V(G)$ is covered, the collection $\{I_\lambda(\alpha)\}_{\lambda \in \Lambda}$ forms a valid graph coloring, requiring $N \ge \chi(G)$.
\end{itemize} 
We therefore have,
\begin{equation}
N \ge \min \left( \chi(G), \max_{\alpha \in (0,1]} \left[ \frac{1 - \epsilon_{\rm c} - \epsilon_{\rm s}^{1-\alpha}}{\epsilon_{\rm s}^\alpha} \right] \right).
\end{equation}
Maximizing over  $\alpha \in (0, 1]$ and simplifying, we obtain 
$$N \ge \begin{cases} 
\chi(G) & \text{if } \epsilon_{\text{s}} \le \frac{(1 - \epsilon_{\text{c}})^2}{4\chi(G)}, \\
\frac{(1 - \epsilon_{\text{c}})^2}{4\epsilon_{\text{s}}} & \text{if } \frac{(1 - \epsilon_{\text{c}})^2}{4\chi(G)} < \epsilon_{\text{s}} < \frac{1 - \epsilon_{\text{c}}}{2}, \\
1 - \epsilon_{\text{c}} - \epsilon_{\text{s}} & \text{if } \epsilon_{\text{s}} \ge \frac{1 - \epsilon_{\text{c}}}{2.},
\end{cases}$$ which yields
Theorem 3 of the main text.

\subsection{6. Entropic Memory Cost and the Fractional Chromatic Number}\label{sec:fractional_chromatic}We now prove that when a classical automaton utilizes shared randomness, its memory overhead for sound solution to the KSP, under an adversarial input distribution, is bounded by the fractional chromatic number $\chi_f(G)$ using an information-theoretic metric. 

Let a classical automaton with shared randomness be defined as a convex mixture of deterministic finite automata $\{M_r\}$, where the shared random parameter $r$ is sampled from a probability distribution $p(r)$ over a compact space $R$ prior to receiving the input. 

To maintain perfect soundness, each individual machine $M_r$ must map the input vertices $V$ to an internal state space $\Lambda$ such that adjacent vertices never share a state. Formally, for any choice of $r$, if $(u, v) \in E(G)$, then $M_r(u) \neq M_r(v)$. Consequently, for a fixed $r$, the preimage of any state $\lambda \in \Lambda$, given by:
\begin{equation}I_{r, \lambda} = \{v \in V \mid M_r(v) = \lambda\},\end{equation} must form a valid independent set of $G$.

Now, let an adversarial opponent select an input vertex $v \in V$ according to a probability distribution $P_V(v)$ of their choice. Given an input $v$, the internal state $S$ of the composite automaton becomes a random variable whose conditional probability distribution is governed by the shared randomness:
\begin{equation}
    P(S = \lambda \mid V = v) = \sum_{r} p(r) \delta(M_r(v), \lambda).
\end{equation}
The unconditional probability distribution of the automaton occupying state $\lambda$ across the protocol is: 
\begin{equation}
    P_S(\lambda) = \sum_{v \in V} P_V(v) P(S = \lambda \mid V = v).
\end{equation}  
Instead of measuring the number of internal states of the automaton (which is still bounded from below by the integer chromatic number $\chi(G)$ for each value of $r$), we define the \textit{entropic state complexity} via the Shannon entropy of the state space:
\begin{equation}
    H(S) = -\sum_{\lambda \in \Lambda} P_S(\lambda) \log_2 P_S(\lambda).
\end{equation}
By K\"orner's graph source coding theorem~\cite{Korner_1973}, the Shannon entropy of any randomized representation where states correspond to independent sets is lower-bounded by the K\"orner graph entropy $H(G, P_V)$ of the input distribution:
$$H(S) \ge H(G, P_V),$$ where $H(G, P_V)$ is defined via the vertex-packing polytope $VP(G)$ as:
\begin{equation}
    H(G, P_V) = \min_{a \in VP(G)} \sum_{v \in V} -P_V(v) \log_2 a_v.
\end{equation}
We allow the adversarial opponent to choose the worst-case input distribution $P_V^*(v)$ that maximizes the above entropy. Under this minimax adversarial formulation, a theorem by Csisz\'ar, K\"orner, Lov\'asz, Marton, and Simonyi~\cite{Csiszar_1990, Simonyi_1995} establishes that the maximum graph entropy over all possible vertex distributions is explicitly given by the fractional chromatic number:
\begin{equation}\max_{P_V} H(G, P_V) = \log_2 \chi_f(G).\end{equation}
Therefore, the worst-case entropic state-complexity satisfies:
\begin{equation}\max_{P_V} H(S) \ge \log_2 \chi_f(G).\end{equation}

For the pure-state QFA solution to the KSP presented in the main text, the corresponding worst-case entropic quantum state complexity remains $\log_2 \xi(G)$. However, if one were to consider a more general QFA architecture -- where the quantum memory is allowed to be a mixed state and the transition functions are arbitrary completely positive trace-preserving (CPTP) maps -- we expect that the worst-case entropic quantum state complexity would be bounded by the Lov\'asz number, $\log_2 \vartheta(G)$~\cite{Duan2013, Boreland_2021, Boreland_2022}.

\section{Boolean-Orthogonality Graphs}
\label{sec:exponential_advantage}

To translate representational contextuality into an exponential quantum memory advantage, we examine the family of \textit{Boolean-orthogonality graphs}, denoted $\Omega_n$~\cite{Buhrman1998,Cameron2007}. For an integer $n$ that is a multiple of 4, the vertex set $V$ of $\Omega_n$ comprises $2^{n-1}$ binary vectors of length $n$ containing an even number of $1$'s. Two vertices share an edge if and only if their Hamming distance is exactly $n/2$. (See Fig.~\ref{fig:ortho_graph} for an illustration of  $\Omega_4$.)

\begin{figure}[!htbp]
    \centering
    \includegraphics[width=0.6\columnwidth]{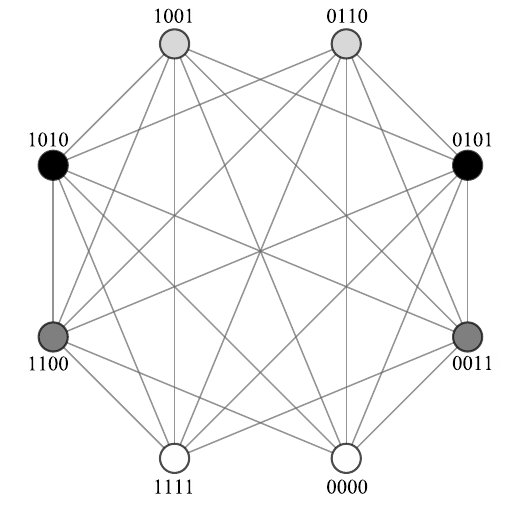}
    \caption{The complete Boolean-orthogonality graph $\Omega_4$. Vertices are colored such that no two connected vertices share a color, illustrating that $\chi(\Omega_4)=4$.}
    \label{fig:ortho_graph}
\end{figure}

An exact orthogonal representation of $\Omega_n$ is constructed by mapping each binary vector coordinate to a real unit vector in $\mathbb{R}^n$ via the substitution $0 \to 1/\sqrt{n}$ and $1 \to -1/\sqrt{n}$. If the Hamming distance between two binary strings is $n/2$, their corresponding vectors in $\mathbb{R}^n$ are orthogonal, yielding an orthogonal rank of $\xi(\Omega_n) = n$. 

Conversely, the classical memory bound is dictated by the chromatic number $\chi(\Omega_n)$. The Frankl-R\"odl theorem~\cite{FranklRodl1987, godsil2008coloring} bounds the independence number $\alpha(\Omega_n)$ to be exponentially small relative to the total vertex size: $\alpha(\Omega_n) \le (2-\delta)^n$ for a positive constant $\delta \approx 0.05$. Because $\chi(G) \ge \chi_F(G) \ge |V|/\alpha(G)$, and given the total vertex count $|V| = 2^{n-1}$, we arrive at the lower bound:
\begin{equation}
\chi(\Omega_n) \geq \chi_F(\Omega_n)\geq\frac{2^{n-1}}{(2 - \delta)^n} = 2^{\Omega(n)}.
\end{equation}
Evaluating the KSP on $\Omega_n$ therefore establishes a task where a absolutely-sound classical machine requires $N = 2^{\Omega(n)}$ states, while the quantum automaton requires only $d = \mathcal{O}(n)$ dimensions (Fig.~\ref{fig:scaling_plot}). This gap applies to the entropic state complexity, i.e., $\chi_f(G)$, as well.

\begin{figure}[!htbp]
    \centering
    \includegraphics[width=0.8\columnwidth]{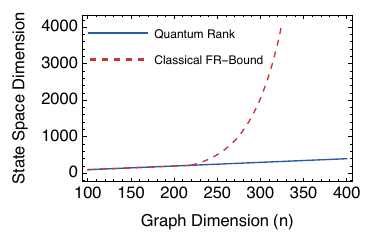}
    \caption{Scaling of the state space dimension as a function of the graph dimension $n$. While the quantum memory cost $\xi(\Omega_n)$ (blue) grows linearly, the lower bound on the classical memory cost $\chi(\Omega_n)$ (red dashed), derived from the Frankl-Rodl (FR) theorem~\cite{FranklRodl1987}, grows exponentially.}
    \label{fig:scaling_plot}
\end{figure}

Let us conclude by addressing a potential concern about the large alphabet size. Although the alphabet size scales exponentially as $2^{n-1}$, the transition unitaries $\{U_v\}$ can be dynamically synthesized in real time via a quantum circuit of size $\mathcal{O}(n)$ using standard techniques. Consequently, the classical controller requires no exponential lookup table and only $\mathcal{O}(n)$ active memory to process each input symbol, preserving the unconditional memory separation.

\section{ Fingerprinting via Expander Graphs}
\label{sec:expander_fingerprinting}

When the soundness error matches or exceeds the critical threshold $\epsilon_{\rm s} \ge 1/\chi(G)$, the classical memory bound no longer applies. Here we provide an explicit PFA for the KSP on Boolean orthogonality graphs, demonstrating that the state complexity collapses from exponential scaling to linear scaling, $\mathcal{O}(n)$. 

Let each vertex $u \in V(\Omega_n)$ be represented by its $n$-bit coordinate string. We encode this string via an asymptotically good linear error-correcting code $C: \{0,1\}^n \to \{0,1\}^m$, where $m = c \cdot n$ for a constant $c > 1$, and the code possesses a minimum relative Hamming distance $\delta > 0$. For any two distinct vertices $u \neq v$, their codewords $C(u)$ and $C(v)$ must disagree on a subset of indices $B \subseteq \{1, \ldots, m\}$ of size $|B| \ge \delta m$.

Let $H$ be a $d$-regular Ramanujan expander graph whose vertex set corresponds to the codeword indices $\{1, \ldots, m\}$. Since $H$ is a Ramanujan graph, the second largest absolute eigenvalue of its normalized adjacency matrix is bounded by the Ramanujan bound $\lambda_2 \le 2\sqrt{d-1}/d$~\cite{Alon_1986, Lubotzky_1988, Margulis_1988, Hoory_2006}. The PFA protocol executes as follows:
\begin{enumerate}
    \item \textbf{Random-Walk Sampling:} The PFA uses internal randomness to sample a uniform random path $W = (x_1, x_2, \ldots, x_{k+1})$ of fixed length $k$ on $H$.
    \item \textbf{First Symbol Transition:} Upon reading $u$, the PFA transitions into an internal memory state uniquely labeled by the path and sampled bit values:
    \begin{equation}
    s = \left(W, C(u)_{x_1}, C(u)_{x_2}, \ldots, C(u)_{x_{k+1}}\right).
    \end{equation}
    \item \textbf{Second Symbol Transition:} Upon reading $v$, the PFA evaluates $C(v)_{x_i}$ along the identical path $W$, accepting if and only if $C(v)_{x_i} = C(u)_{x_i}$ for all $i$.
\end{enumerate}

\paragraph{State Complexity Scaling:} The total number of unique paths of length $k$ starting from an arbitrary vertex in $H$ is $m \cdot d^k$. Recording $k+1$ bits yields $2^{k+1}$ configurations per path. Since $m = c \cdot n$, the total internal classical state space size is:
\begin{equation}
N_{\text{expander PFA}} = m \cdot d^k \cdot 2^{k+1} = \left(c \cdot d^k 2^{k+1}\right) \cdot n = \mathcal{O}(n).
\end{equation}
Because the vertex degree $d$ and path length $k$ are fixed constants independent of $n$, the classical memory scales linearly with $n$, matching the asymptotic quantum scaling $d = \xi(G) = \mathcal{O}(n)$.

\paragraph{Error Analysis:} We verify that this linear-state machine meets the bounded-error promise criteria:
\begin{itemize}
    \item \textbf{Completeness:} If $u = v$, then $C(u) = C(v)$, yielding identical bit evaluations along every step of any path $W$, ensuring perfect completeness ($\epsilon_{\rm c} = 0$).
    \item \textbf{Soundness:} If $u \neq v$, the protocol falsely accepts if and only if the random walk $W$ avoids the disagreement set $B$ entirely, remaining trapped within the agreement set $A = \{1, \dots, m\} \setminus B$. The density of the agreement set is bounded by $\rho = |A|/m \le 1 - \delta$. By the expander walk sampling theorem~\cite{Ajtai1987}, the false-positive soundness error is bounded by:
    \begin{equation}
    \epsilon_{\rm s} \le \left( \rho + \lambda_2 \right)^k \le \left( 1 - \delta + \frac{2\sqrt{d-1}}{d} \right)^k.
    \end{equation}
\end{itemize}
By choosing a sufficiently large constant degree $d$, the spectral gap can be expanded to ensure $\lambda_2 < \delta$, forcing the base of the exponent $|(\rho+\lambda_2)|<1$. Consequently, by choosing the walk length $k$ to be sufficiently large, the soundness error can be driven below any $\mathcal O(1)$ threshold.

\section{Noise Robustness of the Quantum Finite Automaton}
\label{sec:noise_robustness}

As discussed in the main text, we now prove that the QFA solution to the KSP  maintains a constant $\mathcal{O}(1)$ noise threshold independent of the graph's size or chromatic number $N$. Because our KSP construction employs discrete Householder reflections evaluated over a finite input length of $L=2$, continuous unitary errors do not accumulate over scaling string lengths.

\begin{theorem}Let $M_q$ be the $d$-dimensional QFA constructed for the KSP. In the presence of physical noise, the QFA continues to solve the promise problem with a positive acceptance gap $\Delta P > 0$, under the following $N$-independent noise thresholds:
\begin{enumerate}
    \item[(a)] \textbf{Depolarizing noise}: \textit{For any uniform depolarizing rate $p < 1$ per cycle, the gap remains $\Delta P = (1-p)^2 > 0$.}
    \item[(b)] \textbf{Coherent noise}: \textit{For any systematic unitary perturbation inducing an angular deviation bounded by $\epsilon < \pi/8$ per cycle, the gap remains $\Delta P \ge \cos(4\epsilon) > 0$.}
\end{enumerate}
\end{theorem}
\begin{proof} \textbf{Part (a): Depolarizing Noise.} Let the QFA undergo a uniform depolarizing channel $\mathcal{E}(\rho) = (1-p)\rho + p I/d$ per cycle. Applying the channel as the QFA reads the length-$2$ input string yields the final state:
$$\rho_2 = (1-p)^2 U_{v_B}U_{v_A}\ket{\psi_0}\bra{\psi_0}U_{v_A}^\dagger U_{v_B}^\dagger + \left[1-(1-p)^2\right]\frac{I}{d}.$$
Evaluating the trace against the rank-1 acceptance projector $\hat{P}_{\rm accept} = \ket{\psi_0}\bra{\psi_0}$ yields the  valid and invalid acceptance probabilities:
$$\begin{aligned} P_{\rm accept}(v_A v_A) &= (1-p)^2 + \frac{1-(1-p)^2}{d}, \\ P_{\rm accept}(v_A v_B) &= \frac{1-(1-p)^2}{d}. \end{aligned}$$
The dimensional noise background cancels upon subtraction, yielding an invariant acceptance gap of $\Delta P = (1-p)^2 > 0$ for all $p < 1$, which is independent of $N$ and $d$.

\textbf{Part (b): Coherent Noise.} Let the systematic perturbation introduce a unitary drift per cycle with the angular deviation bounded by $\epsilon$. Utilizing the standard angular distance metric $D(\ket{\phi},\ket{\chi}) = \arccos(|\braket{\phi|\chi}|)$, the triangle inequality restricts the maximum state vector drift after two operational cycles to exactly $2\epsilon$. The valid acceptance probability is lower-bounded by $P_{\rm accept}(v_A v_A) \ge \cos^2(2\epsilon)$, while the invalid acceptance probability is upper-bounded by $P_{\rm accept}(v_A v_B) \le \sin^2(2\epsilon)$. A positive acceptance gap requires:
$$\Delta P \ge \cos^2(2\epsilon) - \sin^2(2\epsilon) = \cos(4\epsilon) > 0,$$
establishing a size-independent coherent noise threshold of $\epsilon < \pi/8$. \end{proof}

\section{The 60-Vertex Waegell-Aravind Graph}
\label{app:SM}

To demonstrate near-term experimental viability, we utilize the 60-vertex Kochen-Specker graph $G_{\text{WA}}$ constructed by Waegell and Aravind~\cite{Waegell_2010}, which satisfies $\chi(G_{\text{WA}}) - \xi(G_{\text{WA}}) = 2$ at a minimal orthogonal rank of $\xi(G_{\text{WA}}) = 4$. The vertices are constructed from rays corresponding to the vertices of the 4D regular 600-cell polytope. Let $\phi = \frac{1+\sqrt{5}}{2}$ denote the golden ratio. The 120 vertices of the 600-cell in $\mathbb{R}^4$ are partitioned into three coordinate sets:
\begin{enumerate}
    \item \textbf{Set 1 ($g_1$):} All permutations of $(\pm 2, 0, 0, 0)$, yielding 8 vertices.
    \item \textbf{Set 2 ($g_2$):} All sign combinations of $(\pm 1, \pm 1, \pm 1, \pm 1)$, yielding 16 vertices.
    \item \textbf{Set 3 ($g_3$):} All even permutations of $(0, \pm 1, \pm \phi, \pm \phi^{-1})$, yielding 96 vertices.
\end{enumerate}

Identifying antipodal vectors $v \sim -v$ reduces the set to exactly $|V| = 60$ unique projective rays, each defining a state vector $\ket{v}$ of a single 4-level qudit system. Undirected edges link vertices if and only if their state vectors are orthogonal, $\langle v_i | v_j \rangle = 0$. By construction, therefore, $\xi(G_{\text{WA}}) = 4$.

\begin{figure}[!htbp]
    \centering
    \includegraphics[width=0.8\linewidth]{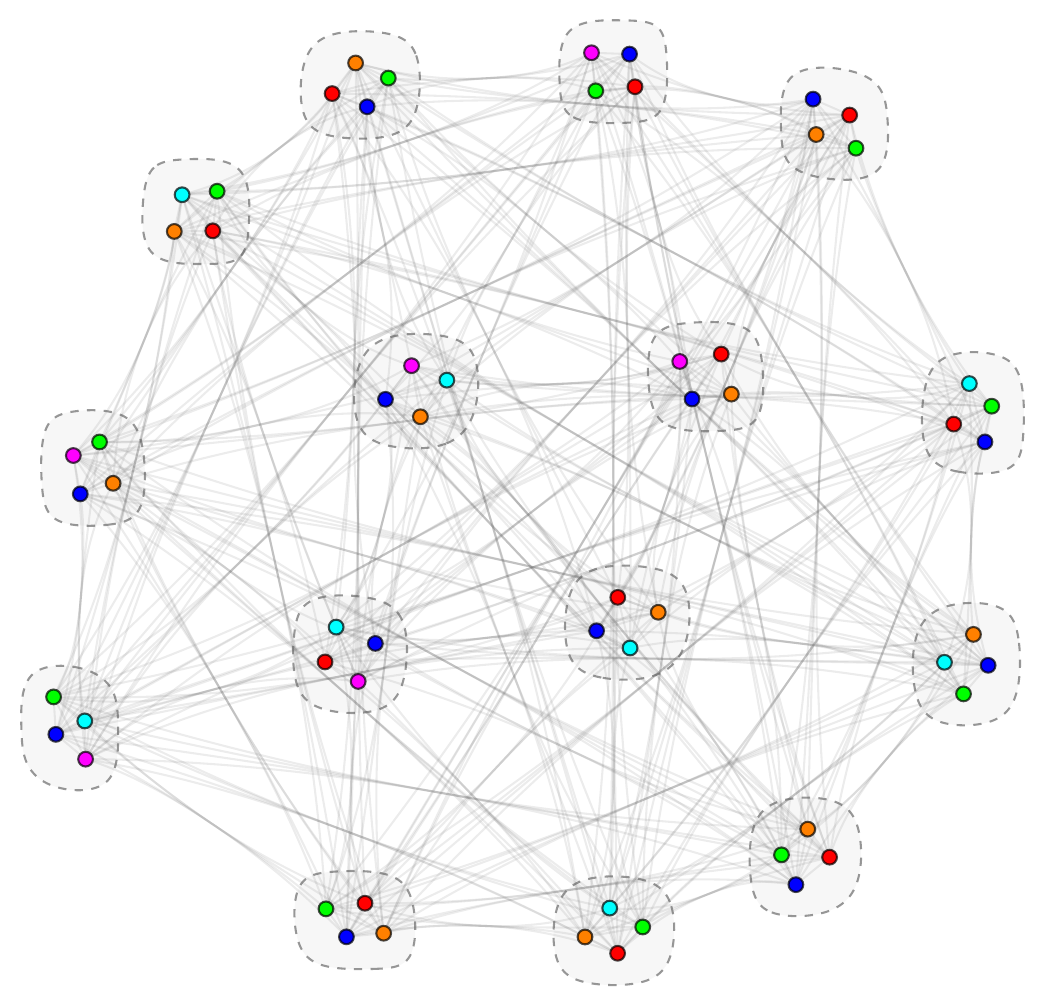}
    \caption{Valid 6-coloring of the 60-ray Waegell-Aravind orthogonality graph~\cite{Waegell_2010}. The graph is partitioned into 15 disjoint 4-element orthogonal bases (shaded regions). Because each basis forms a clique, every region contains four distinct vertex colors.}
    \label{fig:60-v}
\end{figure}

We computed the chromatic number of $G_{\text{WA}}$ by evaluating its independence number to be $\alpha(G_{\text{WA}}) = 13$. The fractional chromatic inequality implies $\chi(G_{\text{WA}}) \geq \lceil |V|/\alpha(G_{\text{WA}}) \rceil = \lceil 60/13 \rceil = 5$. Mapping the graph coloring to a Boolean satisfiability (SAT) problem, an exhaustive search revealed that $G_{\text{WA}}$ is uncolorable using 5 colors. An explicit coloring was found using 6 colors (Fig.~\ref{fig:60-v}), proving that $\chi(G_{\text{WA}}) = 6$. Solving the KSP on this graph requires a classical machine operating at zero soundness error to utilize 6 distinct states, whereas a 4-level quantum system solves the automaton natively using 4 basis states.

\clearpage